\documentclass[a4paper,11pt]{article}
\usepackage{amsmath,amssymb,amsfonts,amsthm}
\usepackage{graphicx}
\usepackage{algorithmic}
\usepackage{algorithm}
\usepackage{setspace}
\usepackage{rotating}
\usepackage{url}
\newtheorem{theorem}{Theorem}
\newtheorem{prop}{Proposition}

\newtheorem{definition}{Definition}

\usepackage{xcolor}
\usepackage{blkarray}

\newcommand\blfootnote[1]{%
  \begingroup
  \renewcommand\thefootnote{}\footnote{#1}%
  \addtocounter{footnote}{-1}%
  \endgroup
}

\usepackage[colorlinks=true, allcolors=blue!65!black]{hyperref}

\begin{document}
\title{Computing the probability of gene trees concordant with the species tree in the multispecies coalescent}


\author{Jakub Truszkowski$^{a,d}$, Celine Scornavacca$^{b,c}$ and Fabio Pardi$^{a,c}$}
\maketitle
\blfootnote{ \hspace{-0.7cm} 
$^a$ LIRMM, CNRS, Universit\'e Montpellier, Montpellier, France \\
$^b$ ISEM, CNRS, Universit\'e Montpellier, Montpellier, France \\
$^c$ Institut de Biologie Computationnelle, Montpellier, France \\
$^d$ Present address: RBC Borealis AI, Waterloo, Ontario, Canada}
\doublespace

\begin{abstract}

The multispecies coalescent process models the genealogical relationships of genes sampled from several species,
enabling useful predictions about phenomena such as the discordance between a gene tree and the species phylogeny due to incomplete lineage sorting.
Conversely, knowledge of large collections of gene trees can inform us about several aspects of the species phylogeny, 
such as its topology and ancestral population sizes. 
A fundamental open problem in this context is how to efficiently compute the probability of a gene tree topology, given the species phylogeny. 
Although a number of algorithms for this task have been proposed, they either produce approximate results, or, when they are exact, 
they do not scale to large data sets.
In this paper, we present some progress towards exact and efficient computation of the probability of a gene tree topology.
We provide a new algorithm that, given a species tree and the number of genes sampled for each species, calculates 
the probability that the gene tree topology will be concordant with the species tree.
Moreover, we provide an algorithm that computes the probability of any specific gene tree topology concordant with the species tree.
Both algorithms run in polynomial time and have been implemented in Python.
Experiments show that they are able to analyse data sets 
where thousands of genes are sampled in a matter of minutes to hours.

\end{abstract}

\section{Introduction} 

Many phenomena may cause a phylogeny for a collection of genes sampled across different species (the {\it gene tree})  to be discordant with the phylogeny of those species (the {\it species tree}) \cite{maddison1997gene}. Sources of discordance include biological processes that cause genes to cross species boundaries, such as horizontal transfer and hybridization, and reconstruction artefacts, such as tree estimation error and the inclusion of paralogous genes. Another crucial cause for this discordance is the randomness in the process of gene lineage coalescence, which may result in gene lineages joining further in the past than the last common ancestor of the species from which they originate ({\it incomplete lineage sorting}).

%

The {\it multispecies coalescent} is the extension of the coalescent model that describes the random process of gene lineages merging within multiple species or populations related by a species tree \cite{hudson1983testing,pamilo1988relationships,takahata1989gene,degnan2009gene}. While tracing gene lineages backwards in time, each species ``acquires'' a number of lineages from its direct descendant species, and ``passes on'' to its direct ancestor the lineages that reach the origin of that species (see Figure \ref{fig-rch-pcoal}).
The process of coalescence inside each species may be dependent on parameters such as the species' effective population size and the ages of speciation events. 
If one is only interested in the branching pattern (topology) of the gene lineages, the only parameters that matter are the branch lengths of the species tree in coalescent units. 

The present paper deals with two key computational challenges arising in the context of the multispecies coalescent. The first is related to the issue of discordance between gene trees and species trees discussed above: 
given a species tree, \emph{how can we compute the probability of observing a gene tree that is discordant} (or, equivalently, concordant) \emph{with that species tree?}
{In the special case where just one lineage is sampled from each leaf of the species
tree, analytical formulae for species trees with up to 5 taxa were {obtained} early on
\cite{hudson1983testing, nei1986stochastic,pamilo1988relationships}, and an algorithm by Degnan and Salter \cite{DegnanSalter} can also be used 
to compute this probability, but {this algorithm} relies on the enumeration of a potentially huge number of scenarios, thus limiting its applicability (more on this below).
As for the general case where multiple lineages can be sampled, a number of alternative definitions of concordance have been proposed \cite{pamilo1988relationships, takahata1989gene, RosenbergAgreement}.} 
In this paper, we adopt the definition of {\it monophyletic concordance}, which 
is the only definition that solely depends on the topology of the gene tree. 
It holds when
all the gene lineages sampled from any 
species form a distinct clade in the gene tree, and the subtree connecting the roots of these clades has the same topology as the species tree. 
In Section \ref{ssec:concordance} we take a closer look at this and other definitions of gene tree/species tree concordance. 
Our main contribution to this first challenge is a general algorithm to compute the probability, under the multispecies coalescent, of generating a gene tree that is monophyletically concordant with a given species tree. 
The running time of our algorithm scales polynomially with the size of the species tree and with the number of sampled lineages.


The second challenge that we tackle in this paper is the following: \emph{how can we efficiently compute the probability of a specific gene tree topology, given a species tree with branch lengths?} This problem has received a lot of attention in the literature, especially since gene tree probability calculations form the basis for some recent methods of species tree inference based on the multispecies coalescent, such as MP-EST \cite{liu2010mp-est} and STELLS \cite{STELLS,pei2017stells2}.

While early works provided solutions for species trees with fewer than five leaves \cite{takahata1985gene,pamilo1988relationships,RosenbergAgreement}, the first general algorithm for this task was given by Degnan and Salter \cite{DegnanSalter}. Their algorithm is based on the concept of {a} {\it coalescent history}, which specifies, for each coalescence in the gene tree, the branch in the species tree where that coalescence occurs. In order to get the exact gene tree probability, all viable coalescent histories of the gene tree must be enumerated. 
Since {their number can grow super-exponentially in the size of the input trees \cite{disanto2015coalescent},}
this algorithm rapidly becomes inapplicable for just a few species and genes. 
{An extension of this algorithm, which can also deal with species phylogenies accounting for hybridization or other reticulate events, 
is available within the \textsc{PhyloNet} package \cite{yu2012probability}.}

A number of improvements over th{e algorithm by Degnan and Salter} were subsequently proposed by Y.~Wu and colleagues \cite{STELLS,CompactCH,pei2017stells2}. The program STELLS \cite{STELLS} is based on a dynamic programming algorithm to calculate the probability of gene trees which relies on the concept of {an} {\it ancestral configuration}, which specifies the set of gene lineages (i.e.~branches of the gene tree) passed by a species to its direct ancestor. Again, enumeration of all possible ancestral configurations is necessary, 
and their number often grows exponentially in the sizes of the input trees.
Recent research has shown that, in fact, when the species {tree} and gene tree are identical, the number of ancestral configurations grows exponentially 
both for balanced and unbalanced families of trees \cite{disanto2017enumeration}. Exponential growth also holds in expectation when the species/gene
tree is drawn uniformly at random, and even holds when only counting ancestral configurations that are non-equivalent for the purpose of the computation
\cite{disanto2019number}. 
In general, STELLS should not be expected to run in polynomial time in the size of its inputs.

{Recently, Y.~Wu proposed another} algorithm (\textsc{CompactCH}; \cite{CompactCH}) based on {\it compact coalescent histories}, which specify for each species the number of gene lineages passed by the species to its ancestor (or, equivalently, the number of coalescences occurring in that species). Only numbers of gene lineages (or coalescences) are given, but not their identity. 
Since the number of compact histories depends exponentially on the size of the species tree, but not on that of the gene tree, a good aspect of this approach is that if the size of the species tree is fixed to a constant, then the algorithm runs in polynomial time in the size of the gene tree \cite{CompactCH}. However, in practice this approach is only feasible for very small species trees.

Finally, the latest version of STELLS (STELLS2; \cite{pei2017stells2}) relies on an algorithm that calculates an approximate value of the gene tree probability. Although the approximation can give probabilities that are many orders of magnitude smaller than the correct probabilities (and there is no guaranteed maximum ratio between them) simulations show a good correlation between these probabilities, and most importantly the approximation does not seem to affect too negatively the accuracy of the species tree estimation carried out by STELLS2, relative to STELLS \cite{pei2017stells2}. The main advantage of STELLS2 is that its gene tree probability algorithm is much faster than the original STELLS algorithm, and in fact it can be shown to run in polynomial time whenever the gene tree is monophyletically concordant with the species tree.

Our contribution to the problem of computing gene tree probabilities in the multispecies coalescent is a polynomial-time, exact algorithm (that is, unlike STELLS2, not involving any approximation) that computes the probability of any specific gene tree topology that is monophyletically concordant with a given species tree. We provide a full running time analysis of this algorithm showing its efficiency, and confirm on simulated data sets that our Python implementation is much faster than the other available implementations of exact algorithms (STELLS and \textsc{CompactCH}, written in C++).

We note that the gene tree probability problem that we consider here is related but very different from that of calculating the probability density of a gene tree with branch lengths, given a species tree with scaled population sizes and species divergence times, which can be expressed analytically and computed very efficiently \cite{rannala2003bayes}. The difference between the two problems lies in the fact that the branch lengths in the gene tree provide a unique localization of coalescent events in the species tree, and therefore avoid the inconvenience of having to take into account all possible alternative localizations. In practice, divergence times in a gene tree are often highly uncertain due to limited amount of sequence data and the presence of reticulate events. Moreover, estimation of divergence times requires researchers to make assumptions on the relationship between per-branch mutation rates and calendar time, such as the molecular clock, which may not be realistic. Inference approaches that do not require branch lengths allow researchers to circumvent these problems. 

Another possible way to avoid relying on branch lengths is to consider gene tree topologies with ranked internal nodes. Their probability, given a species tree, can be computed in polynomial time \cite{degnan2012probability,stadler2012polynomial}, but again the inference must deal with the high uncertainty in the ranking.

The approach we present here is based on the fact that assuming monophyletic concordance 
allows us to decompose the computation of the probability of a gene tree topology given a 
species tree, into computations on pairs of corresponding subtrees of the gene tree and the species tree.
This is possible because, due to monophyletic concordance, the only coalescences 
that can occur on a particular subtree of the species tree are those that are part 
of a corresponding subtree of the gene tree.
Additionally, we use this idea to compute the probability of monophyletic concordance 
in general for a given species tree, by summing over all the monophyletically concordant gene trees, 
while also preserving the recursive structure of the algorithm. 
Our approach bears some resemblance to the recent work of Mehta {\it et al.}~\cite{mehta2016probability}, who devised an algorithm for the problem of computing the probability of monophyly of two groups of genes given a species tree.

%


\section{Preliminaries}

\subsection{The coalescent process}

The coalescent is a stochastic process that models the shared genealogy of a random sample of individuals from a population~\cite{hein2004gene}. Lineages are traced back in time starting from the present. Each coalescence event merges a randomly chosen pair of lineages into one, which corresponds to the shared ancestor of the coalesced lineages. The waiting time until the next coalescence event follows an exponential distribution with rate ${k \choose 2}$ where $k$ is the number of distinct lineages. Consequently, the unit of time in the coalescent is the expected time for two lineages to coalesce. 
For a thorough treatment of the coalescent and its applications, see Hein, Schierup and Wiuf~\cite{hein2004gene}.

\subsection{Gene trees and species trees}

First, some general notation on trees: we write $V(T)$ to denote the set of nodes of tree $T$ and we let $\mathcal{L}(T)$ and $I(T)$ denote respectively the leaves and internal nodes of $T$. For a set of nodes $X \subseteq V(T)$, we write $lca_T(X)$ for the lowest (i.e., most recent) common ancestor of all nodes in $X$ in $T$. For a node $v \in V(T)$, we write $T_v$ to denote the subtree of $T$ induced by $v$ and all of its descendants. 

A {\it species tree} $S$ is a rooted binary tree, whose nodes we take to represent species---either those from which genes are sampled, or their ancestors. Each species $s\in V(S)$ has a length $\ell_s$ representing a measure of the lifespan of $s$, in coalescent units. Note that in other works, species are identified with the branches of $S$, but the two approaches are equivalent, as each node uniquely identifies the branch above it.
A {\it gene tree} $G$ is a rooted binary tree, whose leaves represent the sampled gene lineages, and whose internal nodes represent their past coalescences.

For each leaf $v$ of $G$, $s(v)$ denotes the leaf of $S$ from which $v$ was sampled. For every internal node $v$ of $G$, we define 
\[
sm(v)=lca_S \left(\left\{ s(u): u \in \mathcal{L}(G_v) \right\}\right)
\]
Informally, $sm(v)$ is the lowest species in $S$ where coalescence $v$ can happen {(see Figure~\ref{fig-rch-pcoal}a))}. Finally, we write $Up(v)$ for the set of all ancestors of vertex $v$ and $Down(v)$ for the set of all its descendants. We adopt the convention that every vertex is its own ancestor and descendant.

A {\it coalescent history} for a gene tree $G$ and species tree $S$ is a mapping $\sigma: I(G) \rightarrow V(S)$ such that $\sigma(v)$ is an ancestor of $\sigma(u)$ whenever $v$ is an ancestor of $u$ in $G$ and $\sigma(v) \in Up(sm(v))$. 
This definition is equivalent to the definition of \emph{valid coalescent history} given 
by Degnan and Salter~\cite{DegnanSalter}. We omit ``valid'' for simplicity.
In short, 
a coalescent history identifies the species 
where each coalescence of $G$ happens. 

Given a coalescent history $\sigma$, a  {\it partial coalescent} for species $s$ and $\sigma$ is a one-to-one mapping 
$\rho_s: \sigma^{-1}(s) \rightarrow \left\{ 1,\ldots, |\sigma^{-1}(s)|\right\}$ such that $\rho_s(v)>\rho_s(u)$ 
whenever $v$ is an ancestor of $u$ in $G$. 
That is, every coalescence $v$ mapped to $s$ by the given coalescent history receives an integer rank $\rho_s(v)$ 
that is consistent with the temporal constraints specified by $G$. 
In other words, the partial coalescent specifies the order of coalescences within species $s$; see Figure~\ref{fig-rch-pcoal}b). 

A \emph{ranked coalescent history} for species tree $S$ and gene tree $G$ is a mapping $f: I(G) \rightarrow V(S) \times \mathbb{N}$ defined as $f(v)=(\sigma(v),\rho_{\sigma(v)}(v))$, where $\sigma$ is a coalescent history for $S$ and $G$, and $\rho_s$ is a partial coalescent for $s$ and $\sigma$. A ranked coalescent history can also be defined as a coalescent history together with a partial coalescent for each species in $S$. In short, it identifies the species in which each coalescence occurred, as well as 
the rank within the species given to each coalescence;
see Figure~\ref{fig-rch-pcoal}a). Note that, unlike related concepts (namely \emph{ranked histories} \cite{degnan2012probability}), 
a ranked coalescent history does not specify the relative order of coalescences in different species that exist at the same time.


\begin{figure}
\includegraphics[height=2.2in]{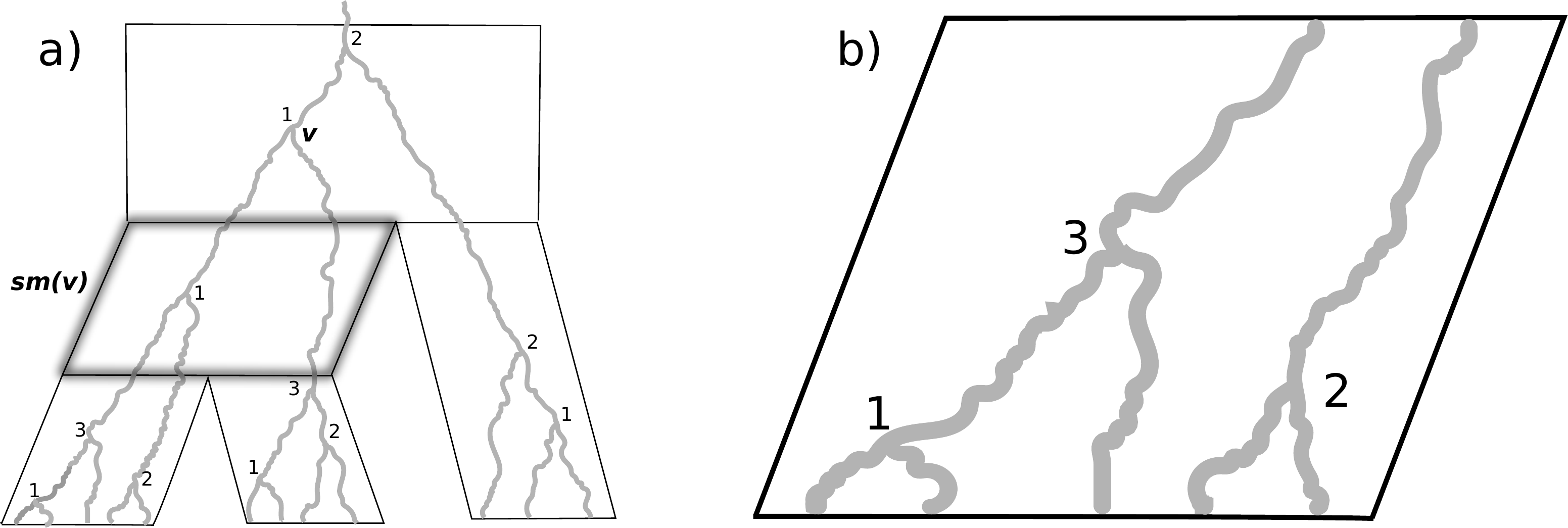}
\caption{a) A ranked coalescent history (grey) within a species tree (black). Within each species (a parallelogram), the coalescent events are numbered from the most recent to the most ancient. For any node $v$ in the gene tree, $sm(v)$ denotes the lowest species where coalescence $v$ can happen. b) A partial coalescent with five lineages at the bottom and two at the top.} 
\label{fig-rch-pcoal} 
\end{figure}

Let $R(S,G)$ be the set of possible ranked coalescent histories for $S$ and $G$. Furthermore, let $R(S,G,k)$ be the set of ranked coalescent histories such that exactly $k$ coalescences of $G$ happen in the root of $S$. 

\subsection{Partial coalescent probabilities}

The probability of a partial coalescent of $s$ depends only on the number of lineages at the bottom and at the top of $s$, as well as the branch length $\ell_s$. It does not depend on the topology of the partial coalescent since, at any point in time, every pair of currently extant lineages is equally likely to be the next one to coalesce. The probability that $n_d$ lineages at the bottom of species $s$ will have $n_u$ ancestors at the top is given by the following well-known formula (see e.g.~\cite{tavare1984line,RosenbergAgreement,DegnanSalter}):
\begin{equation}
p_{n_d,n_u}(\ell_s)=\sum_{k=n_u}^{n_d} e^{-k(k-1)\ell_s/2} \frac{(2k-1)(-1)^{k-n_u}}{n_u!(k-n_u)!(n_u+k-1)} \times \prod_{y=0}^{k-1} \frac{(n_u+y)(n_d-y)}{n_d+y} 
\label{eq-probcoal}
\end{equation}
In practice, we have observed that the above formula is numerically unstable for small values of $\ell_s$. To compute $p_{n_d,n_u}$, we treat the coalescent as a continuous-time Markov process~\cite{TavarePersComm} with an $m \times m$ rate matrix


\[\begin{blockarray}{cccccccc}
\begin{block}{c[cccccc]c}
    &0             &             0  &               0 & 0 & \dots & 0 & {\scriptstyle 1}\\
    &{2 \choose 2} & -{2 \choose 2} &               0 & 0 & \dots & 0 & {\scriptstyle 2}\\
A=	&0             & {3 \choose 2}  & - {3 \choose 2} & 0 & \dots & 0 & {\scriptstyle 3}\\
	& \dots        & \dots          &  \dots      & \dots & \dots & \dots & {\scriptstyle \dots}\\
    &0             &             0  &           \dots & 0 & {m \choose 2} & - {m \choose 2} & {\scriptstyle m}\\	
\end{block}
    &_1 & _2 & _3 & _{\dots} & _{m-1} & _{m} \\
\end{blockarray}\]
where $A_{ij}$ is the transition rate from $i$ to $j$ lineages and $m$ is the number of leaves in the gene tree. Since time is measured in coalescent units, the rate of coalescence from $k$ to $k-1$ lineages is equal to ${k \choose 2}$~\cite{hein2004gene}. The {probability} of transitioning from $n_d$ to $n_u$ lineages within a species of duration $\ell_s$ {is} then obtained by taking the exponent of the rate matrix multiplied by the branch length:
\begin{equation}
p_{n_d,n_u}(\ell_s)=e^{\ell_sA}_{n_dn_u}
\label{eqn:mat_exp}
\end{equation}

The number of partial coalescents with $n_d$ lineages at the bottom and $n_u$ lineages at the top is easily calculated as
\[
w_{n_d,n_u}=\frac{n_d(n_d-1)}{2}\cdot \frac{(n_d-1)(n_d-2)}{2} \cdot \ldots \cdot \frac{(n_u+1)n_u}{2} = \frac{n_d!(n_d-1)!}{n_u!(n_u-1)!} \cdot 2^{-(n_d-n_u)}
\]
Each of these partial coalescents has equal probability, so the probability $L(n_d,n_u,\ell)$ of a specific partial coalescent is simply
\[
L(n_d,n_u,\ell)=p_{n_d,n_u}(\ell)/w_{n_d,n_u}
\]
Conditional on the number of lineages at the bottom and at the top of each species, the coalescent events within each species are independent of those in any other species. Thus, we can write down the probability of a ranked coalescent history $h$ as a product of terms corresponding to each partial coalescent:
\[
\Pr[h|S]=\prod_{s \in V(S)} L(n_{d}(h,s),n_{u}(h,s),\ell_s)
\]
where {$n_{d}(h,s)$ and $n_{u}(h,s)$} are the numbers of lineages at the bottom and top of species $s$ under ranked coalescent history $h$, respectively. Consequently, the probability of a gene tree can be written down as
\[
\Pr[G|S]=\sum_{h \in R(S,G)} \prod_{s \in V(S)} L(n_{d}(h,s),n_{u}(h,s),\ell_s)
\]

Note that whereas Degnan and Salter~\cite{DegnanSalter} decompose $\Pr[G|S]$ into the sum of the probabilities of {\it unranked} coalescent histories for $S$ and $G$, we decompose it into the sum of the probabilities of \emph{ranked} coalescent histories.
Unlike the approach of Degnan and Salter, our formulae do not treat the root species $r$ of $S$ any differently from all other species, 
allowing for the possibility that there remain more than one lineage at the top of $r$, and that no gene tree is generated by the coalescent process. 
In order to guarantee that all lineages eventually coalesce, it suffices to set $\ell_r=+\infty$.

\subsection{Concordance between the gene tree and the species tree \label{ssec:concordance}}


{When} exactly one gene is sampled from each species, concordance between a gene tree $G$ and a species tree $S$ 
means simply that $G$ and $S$ have the same topology. For multiple 
genes per species, defining concordance 
is non-trivial and several different definitions have been proposed. Rosenberg~\cite{RosenbergAgreement} introduced the concept of the {\it collapsed gene tree}, where two species $s_1$ and $s_2$ are considered siblings if the most recent interspecific coalescence in $G$ involves lineages from $s_1$ and $s_2$. The remaining lineages of $s_1$ and $s_2$ are then considered to belong to their shared parent species 
and subsequent nodes of the collapsed gene tree are defined analogously until no interspecific coalescences remain in $G$. Rosenberg then calls $G$ {\it topologically concordant} with the species tree if the collapsed gene tree has the same topology as the species tree. Takahata~\cite{takahata1985gene} proposes a stricter definition of concordance with an additional requirement that the interspecific coalescences in the collapsed gene tree occur in the most recent common ancestral population of the coalescing lineages. 
We refer to Rosenberg's paper  \cite{RosenbergAgreement} for more detailed definitions and discussion.

Both of the above notions of concordance assume the knowledge of the relative time ordering of coalescence events in $G$. Verifying Takahata's concordance between $G$ and $S$ additionally requires knowing the population in which every coalescence in $G$ occurred. This can be a problem for many data sets where divergence time estimates are unreliable. 
Figure \ref{fig-monoconc} illustrates these notions, and Figures \ref{fig-monoconc}b) and \ref{fig-monoconc}c) show an example of a gene tree $G$ with topology 
$\big(\big((A,A),B\big), \big((B,C),(B,C)\big)\big)$ whose concordance with the species tree $((A,B),C)$ (in the sense of Rosenberg or that of Takahata)
depends on a specific ranked coalescent history for $G$ and $S$.


In this paper, we focus on a relatively stringent notion of concordance 
known as {\it monophyletic concordance}, first defined by Rosenberg~\cite{RosenbergAgreement}, which, conveniently, only requires the knowledge of the 
topologies of $G$ and $S$. See Figure~\ref{fig-monoconc}a) for an example.

\begin{definition}
Let $genes(s)$ be the set of leaves in $G$ sampled from species $s$. A gene tree $G$ is monophyletically concordant with $S$ if:
\begin{enumerate} 
\item $genes(s)$ is monophyletic for each leaf species $s$
\item The subtree of $G$ obtained by removing all nodes descendant from $v_s=lca_G(genes(s))$ and labeling 
$v_s$ with $s$, for all leaf species $s$, has the same topology as $S$.  
\end{enumerate}
\end{definition}

\begin{figure}[t]
\includegraphics[scale=0.9]{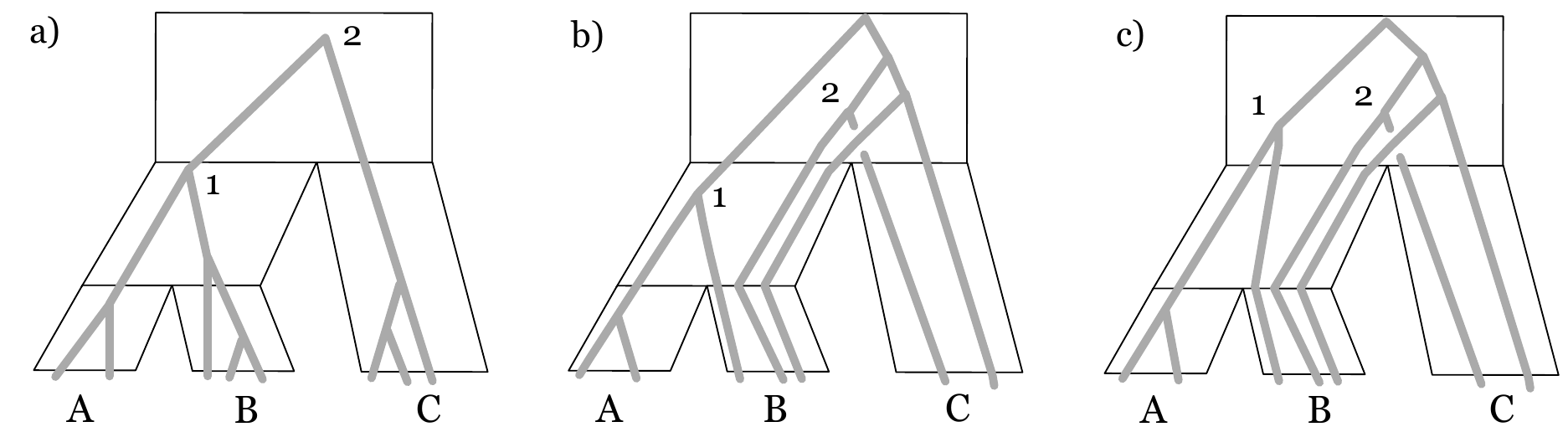}
\caption{Concordance relationships between the gene tree (grey) and the species tree (black parallelograms). 
Only the gene tree in a) is monophyletically concordant with the species tree. 
All three gene trees satisfy Rosenberg's definition of topological concordance, as the most recent interspecific coalescence (labelled 1) is between lineages from A and B, agreeing with the species tree. The gene trees in a) and b) are Takahata-concordant, but not the gene tree in c), as coalescence 1 does not occur in the species directly ancestral to A and B. If we were to invert the relative order of coalescences 1 and 2 in the gene tree in c), then it would not be topologically concordant with the species tree in the sense of Rosenberg. 
We refer to Rosenberg's paper \cite{RosenbergAgreement} for precise definitions of the three notions of concordance.
} 
\label{fig-monoconc}
\end{figure}

If $G$ is monophyletically concordant with $S$, every internal node $s$ in $I(S)$ has a unique node $v \in I(G)$ such that $sm(v)=s$. We write $sm^{-1}(s)$ to denote that unique node. For any leaf node $l$ of $S$, the set of nodes such that $sm(v)=l$ induces a subtree in $G$. We extend the notation above and write $sm^{-1}(l)$ to denote the root of that subtree. Informally, $sm^{-1}(s)$ denotes the most ancient node of $G$ that may have existed in species $s$.




\section{The algorithms}

In the following subsections, we first describe a dynamic programming algorithm for computing the probability of any 
specific gene tree $G$ that is monophyletically concordant with species tree $S$. {
The main idea here is a novel recursion that allows {us} to derive the total probability 
of all possible ranked coalescent histories for species tree $S_s$ and with exactly $k$ coalescences occurring in species $s$,
from corresponding probabilities for the two subtrees $S_{s_1}$ and $S_{s_2}$,
where $s_1$ and $s_2$ are the children of node $s$.
An important observation is that, although our ideas rely on the concept of ranked coalescent histories, 
at no point {does} our algorithm {require} their enumeration. This 
differentiates our approach from previous algorithms for this problem, which {rely} on the 
enumeration of some coalescent scenarios (e.g.~coalescent histories \cite{DegnanSalter}, ancestral configurations 
\cite{STELLS} or compact coalescent histories \cite{CompactCH}).}

{After showing how to deal with the base case of the recursion (Sec.~\ref{sec:base_case}), the key recurrence equation is 
proven in Theorem~\ref{thm:main} (Sec.~\ref{sec:recurrence}). Other details to improve algorithmic efficiency 
are provided in Secs.~\ref{sec:trick}--\ref{sec:multiple}.} 
We then provide an analysis of the running time (Sec.~\ref{sec:runtime}). Finally, we show how the algorithm can be adapted to compute the probability 
that a gene tree is monophyletically concordant with $S$, given $S$ and the number $m_s$ of genes sampled from each leaf species $s$ (Sec.~\ref{sec:alg2}).

\subsection{The single-species case}\label{sec:base_case}

If the species tree consists of a single node $s$ with branch length $\ell_s$, the probability of any gene tree $G$ can be computed efficiently.
Since the probability of any gene tree topology with ranked internal nodes is $L(|\mathcal{L}(G)|,1,\ell_s)$, 
it suffices to multiply that probability by the number of ways of ranking the internal nodes of $G$, which we denote by $r(G)$. 
Tree topologies whose internal nodes are ranked are also known as \emph{labeled histories}~\cite{edwards1970estimation,felsenstein2004inferring}
{or \emph{ranked tree topologies}~\cite{degnan2012probability}}.


For a general rooted binary tree topology $T$, $r(T)$ is given by the following two known formulae, 
where $i_v=|I(T_v)|$, and $c_1(v), c_2(v)$ are the child nodes of $v$ in $T$:
\begin{align}
  r(T) &= \prod_{v\in I(T)} {i_{c_1(v)}+i_{c_2(v)} \choose i_{c_1(v)}} \label{eqn:rT_YangRannala}\\
       &= \frac{i_{root(T)}!}{\prod_{v\in I(T)}i_v} \label{eqn:rT_SteelMcKenzie}
\end{align}
(See, e.g.,~\cite{yang_rannala2014} for Eqn.~\ref{eqn:rT_YangRannala} and~\cite{SteelMcKenzie} for Eqn.~\ref{eqn:rT_SteelMcKenzie}, which is also valid when $T$ is non-binary.)

The probability of the gene tree is then simply
\begin{equation}
\Pr[G|S]=r(G)\cdot L(|\mathcal{L}(G)|,1,\ell_s)
\label{eq-prob-one-species}
\end{equation}


\subsection{The main recursion}\label{sec:recurrence}

If the gene tree $G$ is monophyletically concordant with the species tree $S$, 
then the only coalescences that can occur in a particular clade of the species tree 
themselves form a clade in the gene tree. 
Specifically, the only nodes of $G$ that can appear in $S_s$ are those in $G_{sm^{-1}(s)}$.
(Recall that $T_v$ denotes the subtree of $T$ {containing} $v$ and all of its descendants.)
This observation allows us to subdivide the general problem on $S$ and $G$ into subproblems involving subtrees $S_s$ and $G_{sm^{-1}(s)}$, for all $s\in V(S)$.
More precisely,
we will compute, for each internal node $s$ of $S$, the total probability of all ranked coalescent histories in $R(S_{s},G_{sm^{-1}(s)},k)$ for $0 \leq k \leq |I(G_{sm^{-1}(s)})|$.   
Note that if $G$ were not monophyletically concordant with $S$, then the coalescences that can occur in $S_s$ would not form a single 
clade within $G$, meaning that $sm^{-1}(s)$ could not be uniquely defined.

Let $h$ be a ranked coalescent history in $R(S_{s},G_{sm^{-1}(s)},k)$. For any $j \leq k$, we can modify $h$ by taking the top $j$ coalescence events in $h$ and moving them to the parent of $s$ in $S$. The probability of the modified ranked coalescent history $h'$ is
\[
\Pr[h'|S_{parent(s)}]=\Pr[h|S_s]\cdot\frac{L(j+1,1,\ell_{parent(s)})\cdot L(k+1,j+1,\ell_{s})}{L(k+1,1,\ell_{s})}
\]
As we show below (Theorem \ref{thm:main}), 
we can reason analogously when moving coalescences from two children nodes $s_1$ and $s_2$ to their parent $s$, meaning that the probability
of the resulting ranked coalescent history can be easily expressed as a function of the probabilities of the two ranked coalescent histories from 
$R(S_{s_1},G_{sm^{-1}(s_1)})$ and $R(S_{s_2},G_{sm^{-1}(s_2)})$.

Since any ranked coalescent history in $R(S_{s},G_{sm^{-1}(s)},k)$ can be obtained, in a unique way, by moving $k-1$ coalescences 
(all the ones occurring in $s$ except $sm^{-1}(s)$) from $s_1$ and $s_2$ into $s$, this operation gives us a way to express the total 
probability of $R(S_{s},G_{sm^{-1}(s)},k)$ as a function of analogous probabilities for $s_1$ and $s_2$.


For simplicity, we will slightly abuse the notation and write $i_{s}$ to denote the number of internal nodes in $G_{sm^{-1}(s)}$ to simplify the form of some expressions to follow. In other words, $i_s$ denotes the maximum number of coalescences that can take place in $s$.


\begin{theorem}
Let $P_{s,k}=\sum_{h \in R(S_s,G_{sm^{-1}(s)},k)} \Pr[h|S_s]$.
Then \label{thm:main}
\begin{equation}\hspace{-1cm}
P_{s,k}= \sum_{\substack{k_1',k_2' \geq 0: \\ k_1'+k_2'+1=k}} \sum_{k_1=k_1'}^{i_{s_1}} \sum_{k_2=k_2'}^{i_{s_2}} P_{s_1,k_1}\cdot \frac{L(k_1+1,k_1'+1,\ell_{s_1})}{L(k_1+1,1,\ell_{s_1})} \cdot P_{s_2,k_2}\cdot \frac{L(k_2+1,k_2'+1,\ell_{s_2})}{L(k_2+1,1,\ell_{s_2})} \cdot {k_1'+k_2' \choose k_1'} \cdot L(k+1,1,\ell_s)
\label{eq-psk}
\end{equation}
\end{theorem}

\begin{proof}
Let $h_1 \in R(S_{s_1},G_{sm^{-1}(s_1)},k_1)$ and $h_2 \in R(S_{s_2},G_{sm^{-1}(s_2)},k_2)$. Take $k_1' \leq k_1$ top coalescences from $s_1$ and $k_2' \leq k_2$ top coalescences from $s_2$ and place them in $s$ without changing the ordering of coalescences within $G_{sm^{-1}(s_1)}$ and $G_{sm^{-1}(s_2)}$. Lineages moved from $G_{sm^{-1}(s_1)}$ and $G_{sm^{-1}(s_2)}$ can be interspersed in ${ k_1'+k_2' \choose k_1'}$ ways, meaning that there are ${ k_1'+k_2' \choose k_1'}$ possible ranked coalescent histories that can be created in this way for each choice of $h_1,h_2,k_1'$ and $k_2'$. Each such ranked coalescent history
has probability
\[
\Pr[h|S_s]=\Pr[h_1|S_{s_1}]\Pr[h_2|S_{s_2}]\frac{L(k_1+1,k_1'+1,\ell_{s_1})}{L(k_1+1,1,\ell_{s_1})} \cdot \frac{L(k_2+1,k_2'+1,\ell_{s_2})}{L(k_2+1,1,\ell_{s_2})} \cdot L(k_1'+k_2'+2,1,\ell_s)
\]
where the ratios correspond to the change of partial coalescent probabilities in $s_1$ and $s_2$ due to removing respectively $k_1'$ and $k_2'$ coalescences. The final term denotes the probability of the partial coalescent of $s$; there are $k_1'+k_2'+1$ coalescences in $s$ (the last one being node $sm^{-1}(s)$ in $G$) and all lineages coalesce into one, meaning that there must be $k_1'+k_2'+2=k+1$ lineages at the bottom of $s$. The overall result is obtained from summing over all choices of $h_1,h_2,k_1,k_2$ and $k_1'$.
\end{proof}

We adopt the convention that  $P_{s,k}$ is defined for all $k\in \{0,1,\ldots,i_s\}$, with $P_{s,0}=0$.


\subsection{Computing $P_{s,k}$ efficiently \label{sec:trick}}

The nested sum in Equation~\ref{eq-psk} can be evaluated more efficiently by pre-computing some partial sums. We can rewrite Equation~\ref{eq-psk} as
\begin{equation}\hspace{-1cm}
P_{s,k} {=} {\sum_{\substack{k_1',k_2' \geq 0: \\ k_1'+k_2'+1=k}}} 
\left( \sum_{k_1=k_1'}^{i_{s_1}} P_{s_1,k_1} \frac{L(k_1+1,k_1'+1,\ell_{s_1})}{L(k_1+1,1,\ell_{s_1})} \right) 
\left( \sum_{k_2=k_2'}^{i_{s_2}} P_{s_2,k_2} \frac{L(k_2+1,k_2'+1,\ell_{s_2})}{L(k_2+1,1,\ell_{s_2})} \right) 
\cdot {k_1'+k_2' \choose k_1'} \cdot L(k+1,1,\ell_s)
\label{eq-psk-factor}
\end{equation}
Letting 
\begin{equation}
U_{s,k'}=\sum_{k=k'}^{i_s}P_{s,k}\frac{L(k+1,k'+1,\ell_{s})}{L(k+1,1,\ell_{s})}
\label{eqn:usk}
\end{equation}
we can rewrite Equation~\ref{eq-psk} as
\begin{equation}
P_{s,k}=\sum_{\substack{k_1',k_2'\ge 0: \\ k_1'+k_2'+1=k}} U_{s_1,k_1'}U_{s_2,k_2'} {k_1'+k_2' \choose k_1'} \cdot L(k+1,1,\ell_s)
\label{eq-psk-fast}
\end{equation}
Assuming that all the relevant $L(n_d,n_u,\ell_s)$ values have been previously computed (see Sec.~\ref{sec:runtime}),
as well as the {$U_{s_i,k'_i}$} values for $s_1$ and $s_2$, 
Equation~\ref{eq-psk-fast}
allows us to evaluate each $P_{s,k}$ in $O(k)$ time, 
and thus all $P_{s,k}$'s for a given $s$ in $O(i_s^2)$. 
Similarly, computing all $U_{s,k'}$ values for a given $s$ takes $O(i_s^2)$ time,
using Equation~\ref{eqn:usk}.


\subsection{The complete algorithm}

Our algorithm computes, for each leaf $s$ of $S$, the probability of the subtree of $G$ composed of leaves in $genes(s)$ and their common ancestors using Equation~\ref{eq-prob-one-species}. Then it conducts a post-order traversal of the internal nodes of $S$ and computes $P_{s,k}$ values by repeatedly applying the recurrences described above. 
The gene tree probability is the sum of all the values of $P_{s,k}$ at the root of $S$.
Algorithm~\ref{alg-gene-tree-agree} illustrates this in detail.

\begin{algorithm}
\caption{GeneTreeProb($S$,$G$)}
\label{alg-gene-tree-agree}
\begin{algorithmic}
\FORALL{$s \in \mathcal{L}(S)$}
	\STATE Find $r=lca_G(genes(s))$
	\STATE Compute $\Pr[G_r|S_s]$ using Equation~\ref{eq-prob-one-species}
	\STATE Set $P_{s,i_r}=\Pr[G_r|S_s]$,  $P_{s,k}=0$ for $0 \leq k < i_r$
	\STATE Compute and store $U_{s,k}$ values for $0 \leq k \leq i_r$
\ENDFOR
\FORALL{$s \in I(S)$ in post-order}
    \FORALL{$0 \leq k \leq i_s$}
		\STATE Compute $P_{s,k}$ using Equation~\ref{eq-psk-fast}
        \STATE Compute and store $U_{s,k}$ using Equation~\ref{eqn:usk}
	\ENDFOR
\ENDFOR
\RETURN $\sum_{k=1}^{|I(G)|}P_{root(S),k}$
\end{algorithmic}
\end{algorithm}



\subsection{Analyzing multiple gene trees}\label{sec:multiple}


Algorithm~\ref{alg-gene-tree-agree} should be preceded by the precomputation of all relevant $L(n_d,n_u,\ell_s)$ values, as each of these values may be needed multiple times. This becomes particularly important when the goal is to calculate the probability of multiple gene tree topologies $G_1, \dots, G_t$ that are monophyletically concordant with the same fixed species tree $S$. This task may be of interest, for example, when there is uncertainty in the inferred gene tree topology, and we may want to consider (or sum over) a set of alternative hypotheses.  
Algorithm~\ref{alg-multiple-trees} makes the precomputation of the $L(n_d,n_u,\ell_s)$ values explicit, as well as the possibility of analyzing multiple gene trees.  
We assume that, for any given leaf species $s$, the number of genes sampled from $s$ in $G_k$ is the same for all $k\in \{1,\dots,t\}$ (and therefore $i_s$ is the same for all gene trees). This is a reasonable assumption, for example, if the gene tree topologies are alternative evolutionary hypotheses for the same multiple sequence alignment.

\begin{algorithm}
\caption{GeneTreesProbabilities($S, G_1, \cdots, G_t$)}
\label{alg-multiple-trees}
\begin{algorithmic}
\FORALL{$s \in V(S)$}
    \FORALL{$n_d,n_u$ such that $1\le n_u \le n_d \le i_s + 1$}
        \STATE Compute and store $L(n_d,n_u,\ell_s)$
    \ENDFOR
\ENDFOR
\FORALL{$i$ from $1$ to $t$}
    \STATE GeneTreeProb($S$,$G_i$)
\ENDFOR 
\end{algorithmic}
\end{algorithm}


\subsection{Running time analysis} \label{sec:runtime}

When analyzing a single gene tree, the running time of our algorithm is dominated by computing the $L(n_d,n_u,\ell_s)$ values. 
This is done by exponentiating the matrix $\ell_sA$ (see Eq.~\ref{eqn:mat_exp}) using a numerical algorithm from the Python package SciPy~\cite{AlMohy}, which takes $O(i_s^3)$ time for node $s$. 
In contrast, computing all $U_{s,k}$ and $P_{s,k}$ values for a given $s$ takes $O(i_s^2)$ time (see Section~\ref{sec:trick}). Thus, if $n$ is the number of species in $S$, and $m$ the number of genes (leaves) in $G$, the total time spent computing all $L(n_d,n_u,\ell_s)$ values 
is at most $O(n m^3)$, whereas the total time spent computing $U_{s,k}$ and $P_{s,k}$ is $O(n m^2)$. This yields a total running time, for one gene tree, of $O(n m^3)$.

When computing the probability of multiple gene trees given a fixed species tree, the $L(n_d,n_u,\ell_s)$ values only need to be computed once, whereas the $U_{s,k}$ and $P_{s,k}$ values must be computed separately for each gene tree, as shown in Algorithm~\ref{alg-multiple-trees}. We thus have the following result.

\begin{prop}
Algorithm~\ref{alg-multiple-trees} requires $O(n m^3+t n m^2)$ time to calculate the probability of $t$ gene trees with at most $m$ genes each, with respect to a species tree with $n$ species.
\label{thm:complexity}
\end{prop}


Interestingly, in a number of realistic scenarios where the species tree is symmetric and the gene sampling is uniform, the above worst-case time bound can be improved. Suppose that the same number of genes is sampled from each leaf species, and that $S$ is fully balanced, meaning that for each internal node $s\in I(S)$, the subtrees rooted in its children $s_1$ and $s_2$ have the same number of leaves. Because of the uniform gene sampling, $n = O(m)$, which implies that Proposition~\ref{thm:complexity} would give an upper bound on runtime of $O(m^4+tm^3)$. However, because of the symmetry of the species tree, the running time satisfies
\[
T(m) = 2T(m/2) + O(m^3) + O(tm^2)
\]
where 
$2T(m/2)$ corresponds to the time spent on the two subtrees rooted on the children of the root $r$ of the tree, 
$O(m^3) = O(i_r^3)$ corresponds to the time spent on computing the $L(n_d,n_u,\ell_r)$ values at $r$, 
and $O(tm^2) = O(t i_r^2)$ corresponds to the computation of the $U_{r,k}$ and $P_{r,k}$ values for all gene trees.
Based on the recurrence above, it can be shown using standard algorithm analysis techniques~\cite{Cormen} that $T(m)$ is bounded by $O(m^3+tm^2)$, rather than $O(m^4+tm^3)$. This suggests that for many practical scenarios featuring balanced gene samplings and trees, our algorithm will run an order of magnitude faster than what Proposition~\ref{thm:complexity} suggests.



\subsection{Computing the probability of monophyletic concordance} \label{sec:alg2}


When the number of genes sampled from each species is large, the probability of any gene tree topology will be quite small simply because of the large number of possible partial coalescents in any leaf node of the species tree. In such situations, it might be useful to instead compute the probability that the gene tree is monophyletically concordant, given a species tree and the number $m_s$ of sampled genes from each leaf species $s$. We emphasize that this probability does not depend on a specific gene tree, but instead is the sum of the probabilities of all gene trees that are monophyletically concordant with $S$, given the numbers of sampled genes.



Algorithm~\ref{alg-gene-tree-agree} can be modified to compute that probability by changing the initialization at the leaves of $S$. Let $\mathcal{G}(s)$ be the set of all gene tree topologies monophyletically concordant with $S_s$ with $m_{s'}$ sampled genes from each leaf ${s'}\in \mathcal{L}(S_s)$. We define
\[
P'_{s,k}= \sum_{G' \in \mathcal{G}(s)} \sum_{h \in R(S_s,G',k)} \Pr[h|S_s]
\]
We can view each $P'_{s,k}$ as 
the sum of the values of $P_{s,k}$ over all $G'$ that are monophyletically concordant with $S_s$.
For each leaf $s$, we initialize $P'_{s,m_s-1}$ to $p_{m_s,1}(\ell_s)$. This corresponds to summing over all gene tree topologies over samples from $s$. The trees in $\mathcal{G}(s)$ differ only in the relationships between genes 
sampled from the same leaf species. Note that for each of these gene trees, Equation~\ref{eq-psk} is satisfied. Summing both sides of Equation~\ref{eq-psk} over all $G'\in \mathcal{G}(s)$, it is easy to see that Equation~\ref{eq-psk} also holds with $P'_{s,k}$ replacing $P_{s,k}$.
Thus, we can compute the values of $P'_{s,k}$ recursively in the same way as $P_{s,k}$ following Equations~\ref{eq-psk},~\ref{eqn:usk}, and~\ref{eq-psk-fast}. The only change to the algorithm is the initialization of $P'_{s,k}$ values at the leaves. Algorithm~\ref{alg-mono-concord} illustrates this in detail.

The running time complexity of Algorithm~\ref{alg-mono-concord} is the same as that of the algorithm to compute the probability of a single gene tree topology with $m_s$ genes sampled from leaf species $s$, as the only difference is the initialization step, which carries no extra cost. Thus, Algorithm~\ref{alg-mono-concord} runs in $O(nm^3)$ time, where $m=\sum_{s\in \mathcal{L}(S)} m_s$.

\begin{algorithm}
\caption{MonophyleticConcordance($S$,$\left\{m_s|s \in \mathcal{L}(S) \right\}$)}
\label{alg-mono-concord}
\begin{algorithmic}
\FORALL{$s \in V(S)$ in post-order}
    \STATE \textbf{if} $s$ has children $s_1$ and $s_2$, \textbf{then} let $m_s=m_{s_1}+m_{s_2}$
    \FORALL{$n_d,n_u$ such that $1\le n_u \le n_d \le m_s$}
        \STATE Compute and store $L(n_d,n_u,\ell_s)$
    \ENDFOR
\ENDFOR
\FORALL{$s \in \mathcal{L}(S)$}
	\STATE Set $P'_{s,m_s-1}=p_{m_s,1}(\ell_s)$,  $P'_{s,k}=0$ for $0 \leq k < m_s-1$.
	\STATE Compute and store $U_{s,k}$ values for $0 \leq k \leq m_s-1$.
\ENDFOR
\FORALL{$s \in I(S)$ in post-order}
    \FORALL{$0 \leq k \leq m_s-1$}
		\STATE Compute $P'_{s,k}$ using Equation~\ref{eq-psk-fast}
        \STATE Compute and store $U_{s,k}$ using Equation~\ref{eqn:usk}
	\ENDFOR
\ENDFOR
\RETURN $\sum_{k=1}^{m_{root(S)}}P'_{root(S),k}$
\end{algorithmic}
\end{algorithm}





\section{Implementation and experimental validation}

A Python implementation of all the algorithms described above, named \textsc{gtprob}, 
is available at~\url{https://github.com/truszk/gtprob}. 
As we show below, \textsc{gtprob} is able to process gene trees of up to 1000 taxa in a matter of minutes, on standard desktop computers. 


\label{sec-experiments}

Since exact (but non-polynomial) algorithms for the problem of calculating the probability of a gene tree given a species tree are known, 
we decided to test the component of \textsc{gtprob} that solves the same problem, restricted to monophyletically concordant gene tree topologies (Algorithm~\ref{alg-gene-tree-agree}).
We compared the running times and the probabilities computed by 
\textsc{gtprob} against those of two available programs: 
STELLS~\cite{STELLS} and \textsc{CompactCH}~\cite{CompactCH}. The latter 
is designed to be efficient for small numbers of species and large numbers of sampled genes. 
To investigate the relative advantages of different algorithms, we focused on two scenarios: one with large numbers of leaf species with small numbers of genes sampled from 
each species, and the other with small numbers of leaf species and large numbers of genes. 
Within each scenario, we varied the number of species and genes 
to investigate scalability.

We evaluated the performance of the algorithm on a large set of simulated data sets. 
Here, we define a data set as a single pair $(S,G)$, where the gene tree $G$
is monophyletically concordant with the species tree $S$.
To construct each 
data set, we sampled the species tree 
from the Yule process, conditioned on a specified number of leaves. 
Following that, we sampled a monophyletically concordant gene tree topology. 
First, we sampled a topology 
for each set of genes 
drawn from the same leaf species. 
These topologies were then connected with each other in such a way as to produce a monophyletically concordant gene tree. 
All experiments were run on servers with cores operating at 2.66 GHz and 72 GB of memory, except experiments in Appendix~\ref{sec-examples}, which were run on a laptop with cores operating at 2.3GHz and 8 GB of memory. Only one core was allowed for each run.



\subsection{Agreement between methods}

For almost all data sets we analyzed, all exact algorithms gave virtually the same probability values. For those data sets where we could successfully run STELLS, the difference in log-probability between \textsc{gtprob} and STELLS was less than $10^{-7}$ for the vast majority of data sets, and was never greater than $10^{-6}$. The difference between \textsc{gtprob} and \textsc{CompactCH} was similar, except for instances with 2 leaf species and 50 leaf genes in each species, where \textsc{CompactCH} reported much lower log-probability values, sometimes by 8 orders of magnitude. We suspect that this was caused by a software bug in \textsc{CompactCH} code. 
To investigate these discrepancies, we ran another exact approach, \textsc{PhyloNet}~\cite{yu2012probability}, on a number of small data sets, which always resulted in log-probabilities 
that were much closer to those of \textsc{gtprob} than those of STELLS.
A few instructive examples are presented in Appendix~\ref{sec-examples}. Since \textsc{PhyloNet} is slower than the other exact algorithms, we did not run it on the entirety of our simulated data sets.

We have also compared the output of our program to STELLS2, which is a faster, approximate method of computing gene tree probabilities. For all the data sets we have analyzed, the results of STELLS2 were vastly different from the other algorithms, with log-probabilities usually approximately twice the value given by the exact algorithms. For this reason, we have decided to focus on exact algorithms in our experiments.

\subsection{Running time comparison}\label{sec:runtime_moderate}
In a first set of experiments we compared running times over moderate-size data sets. We varied $n$, the number of species, and $m_s$, the number of genes sampled in each species, over every combination of $n\in \{4,8,16,24,32,40\}$ and $m_s\in \{1,3,5\}$, corresponding to scenarios with many species and few genes per species, and every combination of $n\in \{2,3,4\}$ and $m_s\in \{5,10,15,20,50\}$, for scenarios with few species and many genes. For each experimental setting, we simulated 10 data sets and reported the running times for the three algorithms. 
The average running times 
in the many species/few genes scenario and in the few species/many genes scenario are shown in Tables~\ref{tab-runtimes-sp} and~\ref{tab-runtimes-genes}, respectively.

We see that the running time of our algorithm grows relatively slowly as we increase the number of species or the number of genes. 
As a result, our algorithm is much faster than either STELLS or \textsc{CompactCH} in most settings. The only exceptions are data sets with very few 
genes, 
where STELLS is sometimes faster (top left of Table~\ref{tab-runtimes-sp} and first line of Table~\ref{tab-runtimes-genes}), and data sets with just $2$ species, where \textsc{CompactCH} is faster than our algorithm (left part of Table~\ref{tab-runtimes-genes}). For larger data sets, our algorithm runs much faster than the competing approaches, and is able to process instances that are impossible to analyze with either STELLS or \textsc{CompactCH}. We also observed that the variance in running times, across the 10 replicates, was in general much lower for \textsc{gtprob} than for the other programs (data not shown;  
see also Appendix~\ref{sec-examples}). 



\begin{table}[t!]
\footnotesize
\centering
\begin{tabular}{|c||c|c|c||c|c|c||c|c|c|}
\# species & \multicolumn{3}{c||}{1 gene/species} & \multicolumn{3}{c||}{3 genes/species} & \multicolumn{3}{c|}{5 genes/species} \\
\hline
 	& STELLS & \textsc{c.CH} & \textsc{gtprob} & STELLS & \textsc{c.CH} & \textsc{gtprob} & STELLS & \textsc{c.CH} & \textsc{gtprob} \\
\hline
4	&	{\bf 0.000}	&	{\bf 0.000}	&	0.345 	    &   {\bf 0.004}	&	0.153 	&	0.341 	    & {\bf 0.357}	&	3.025 	&	0.360 	\\
8	&	{\bf 0.000}	&	0.031 	    &	0.353 	    &	8.305 	    &	NA	    &	{\bf 0.423}	&	6003	    &	NA	&	{\bf 0.502}	\\
16	&	{\bf 0.092}	&	251.6 	    &	0.397 	    &	NA	        &	NA	    &	{\bf 0.672}	&	NA	        &	NA	&	{\bf 1.065}	\\
24	&	1.846 	    &	NA	        &	{\bf 0.442}	&	NA	        &	NA	    &	{\bf 0.928}	&	NA	        &	NA	&	{\bf 2.191}	\\
32	&	887.4  	    &	NA	        &	{\bf 0.504}	&	NA	        &	NA	    &	{\bf 1.647}	&	NA	        &	NA	&	{\bf 3.433}	\\
40	&	20852 	    &	NA	        &	{\bf 0.600}	&	NA	        &	NA	    &	{\bf 1.970}	&	NA	        &	NA	&	{\bf 7.415}	\\
\end{tabular}
\caption{Average running times (in seconds) of the three exact algorithms for varying numbers of species. \textsc{c.CH} stands for \textsc{CompactCH}. NA entries indicate cases where none of the 10 runs completed within 24 hours.
The lowest running time for each parameter combination is indicated in bold.}
\label{tab-runtimes-sp}
\bigskip\bigskip
\centering
\footnotesize
\begin{tabular}{|c||c|c|c||c|c|c||c|c|c|}
\# genes/species & \multicolumn{3}{c||}{2 species} & \multicolumn{3}{c||}{3 species} & \multicolumn{3}{c|}{4 species} \\
\hline
 	& STELLS & \textsc{c.CH} & \textsc{gtprob} & STELLS & \textsc{c.CH} & \textsc{gtprob} & STELLS & \textsc{c.CH} & \textsc{gtprob} \\
\hline
5	&	{\bf 0.000}	& {\bf 0.000} &	0.313 &	{\bf 0.019} &	0.063 &	0.342 &	{\bf 0.357}	&	3.025 	&	0.360 	\\
10	&	0.062 &	{\bf 0.005}	&	0.357 	&	12.33	&	1.227	&	{\bf 0.393}	&	2998 &	160.3	&	{\bf 0.458}	\\
15	&	3.774 &	{\bf 0.021}	&	0.367	&	4905	&	6.395	&   {\bf 0.464}	&	NA	&	2204	&	{\bf 0.561}	\\
20	&	144.1 &	{\bf 0.045}	&	0.381	&	NA	    &	21.75	&	{\bf 0.538}	&	NA	&	16879	&	{\bf 0.833}	\\
50	&	NA	  &	{\bf 0.527}	&	0.713 	&	NA	    &	1766 	&	{\bf 1.879}	&	NA	&	NA	    &	{\bf 4.014}	\\
\end{tabular}
\caption{Average running times (in seconds) of the three exact algorithms for varying numbers of genes per species. \textsc{c.CH} stands for \textsc{CompactCH}. NA entries indicate cases where none of the 10 runs completed within 24 hours.
The lowest running time for each parameter combination is indicated in bold.}
\label{tab-runtimes-genes}
\end{table}

\subsection{Scalability of \textsc{gtprob} for large data sets}

Finally, we investigated the performance of \textsc{gtprob} for large data sets that are not within reach of the current algorithms. We varied data set sizes from tens to thousands of genes. In line with the experiments in the previous section, we simulated 10 data sets for each experimental condition. 

The running times are shown in Figure~\ref{fig-runtimes-large}. We see that our algorithm is able to process gene trees with hundreds of leaves in a matter of seconds or minutes. For gene trees with over {$1000$} leaves, the running times range from about an hour to a few days. All in all, our algorithms can process data sets that are around two orders of magnitude larger {than} what is possible with STELLS.

\begin{figure}
\centering
\begin{minipage}{.5\textwidth}
\centering
  \includegraphics[width=1.01\linewidth]{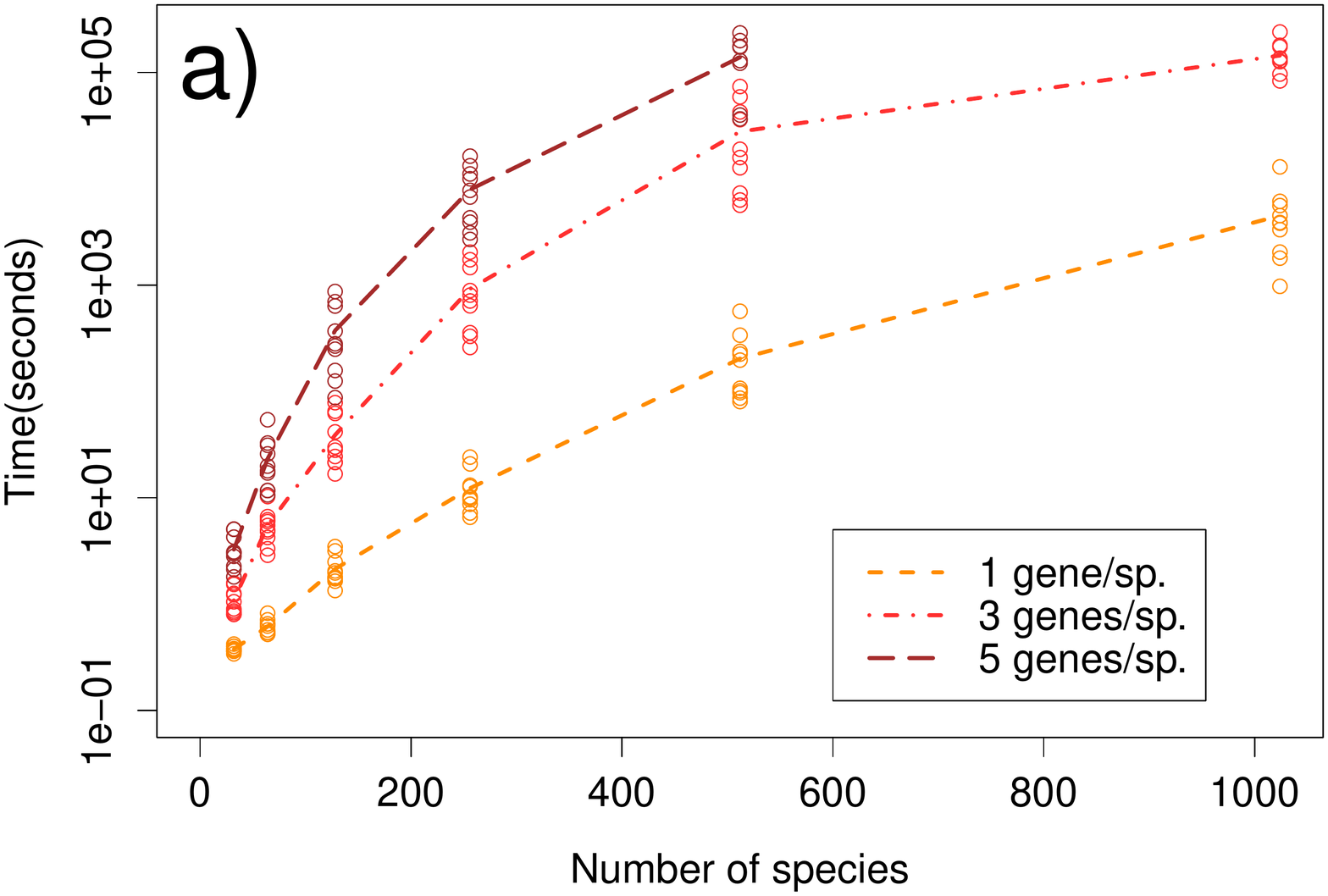}
\end{minipage}%
\begin{minipage}{.5\textwidth}
\centering
  \includegraphics[width=1.01\linewidth]{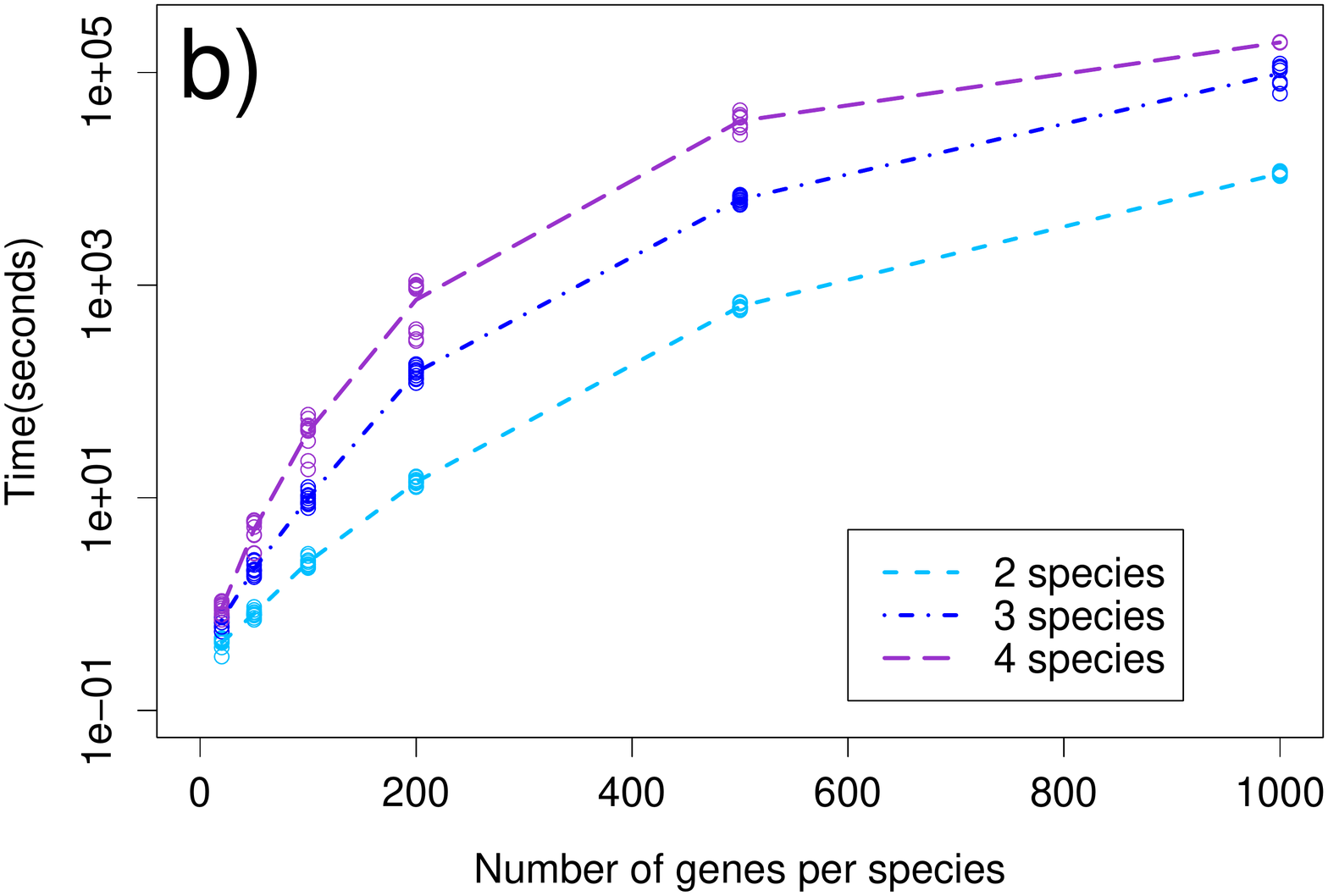}
\end{minipage}
\caption{a) The running times of Algorithm 1 for large numbers of leaf species. b) The running times of Algorithm 1 for large numbers of samples per species.} 
\label{fig-runtimes-large}
\end{figure}

\section{Conclusions}

We have presented efficient, polynomial-time algorithms for two problems involving the multispecies coalescent. 
The first algorithm computes the probability of any specific monophyletically concordant gene tree topology given a species tree.
The second algorithm computes the probability of monophyletic concordance of the gene tree, given a species tree.

From an algorithmic standpoint, these two results 
represent a major improvement over previous exact approaches. 
A clear example of this is provided by the case where only one gene is 
sampled from each leaf species. 
Note that in this case there is only one concordant gene tree topology,
regardless of the definition of concordance.
Computing the probability of this topology, and thus of concordance, 
using any of the previous approaches, involves the enumeration of a
number of scenarios that explodes combinatorially in the size of the species tree. 
Specifically, the maximum number of coalescent histories (enumerated by the algorithm of Degnan 
and Salter \cite{DegnanSalter}) grows super-exponentially \cite{disanto2015coalescent},
while that of ancestral configurations (enumerated by STELLS \cite{STELLS}) and 
compact coalescent histories (enumerated by \textsc{CompactCH} \cite{CompactCH})
grows exponentially.
This combinatorial explosion does not only concern a worst-case analysis, but also holds 
in an average case, for species trees drawn uniformly at random \cite{disanto2017enumeration, disanto2019enumeration, disanto2019number}.

Because our algorithms entirely avoid enumerating any type of coalescent scenario, they also avoid the issue of combinatorial explosion. In fact,
despite their implementation in a relatively inefficient programming language (Python), 
our algorithms enable the analysis of data sets with thousands of samples in a matter of hours, on standard desktop machines. This is roughly two orders of magnitude more than is possible with current algorithms (implemented in C/C++). 

All known exact algorithms for computing the probability of general tree topologies under the multispecies coalescent do not scale beyond tens of genes or several species. We hope that the algorithmic techniques introduced in this work will also lead to more efficient algorithms for computing non-concordant gene tree probabilities. We plan to explore this direction in future work. 

\section{Acknowledgements}

We would like to thank Simon Tavar{\'e} for suggesting the solution to the numerical problems with computing $p_{n_d,n_u}(\ell)$.  
We are also grateful to the associate editor and the reviewers, whose constructive comments helped us to significantly clarify the manuscript.

\bibliographystyle{plain}
\bibliography{refTPB}

\vspace{3cm}

\appendix


\begin{figure}[h!]
\centering
\begin{minipage}{.43\textwidth}
\centering
  \includegraphics[width=0.7\linewidth]{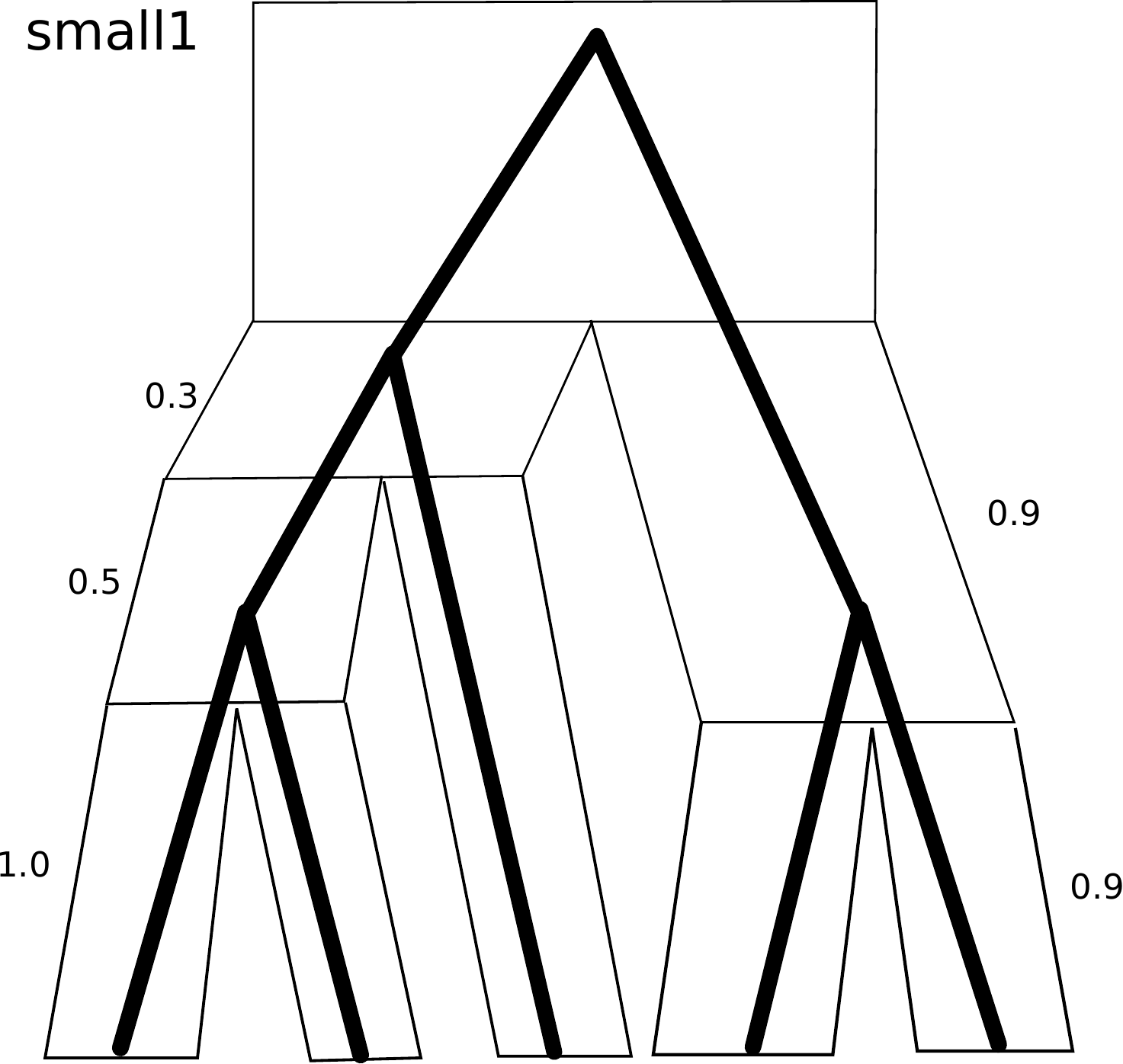}
  \\
  \vspace{0.5cm}
  \includegraphics[width=1.0\linewidth]{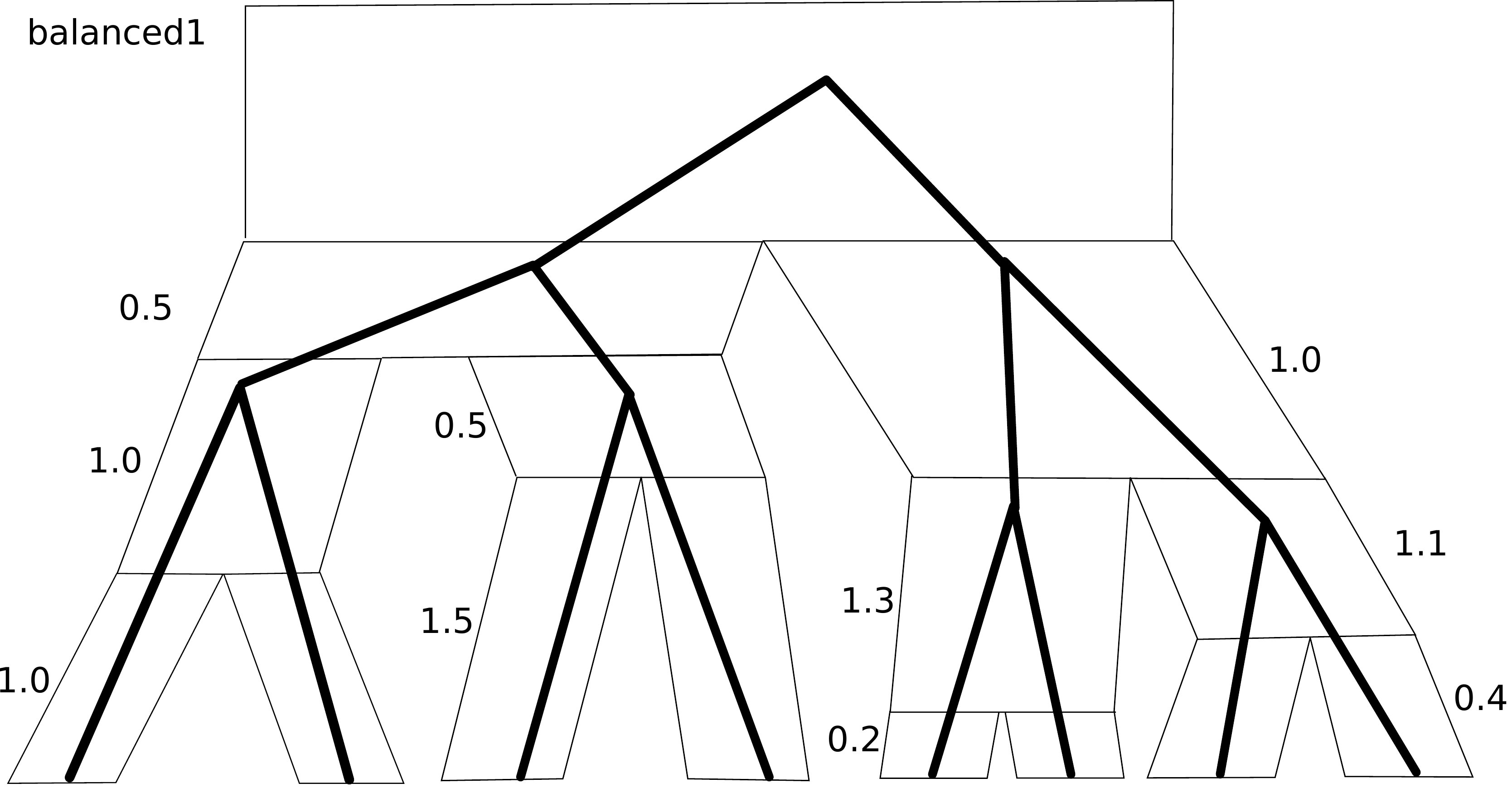}
  \\
  \vspace{0.5cm}
  \includegraphics[width=1.0\linewidth]{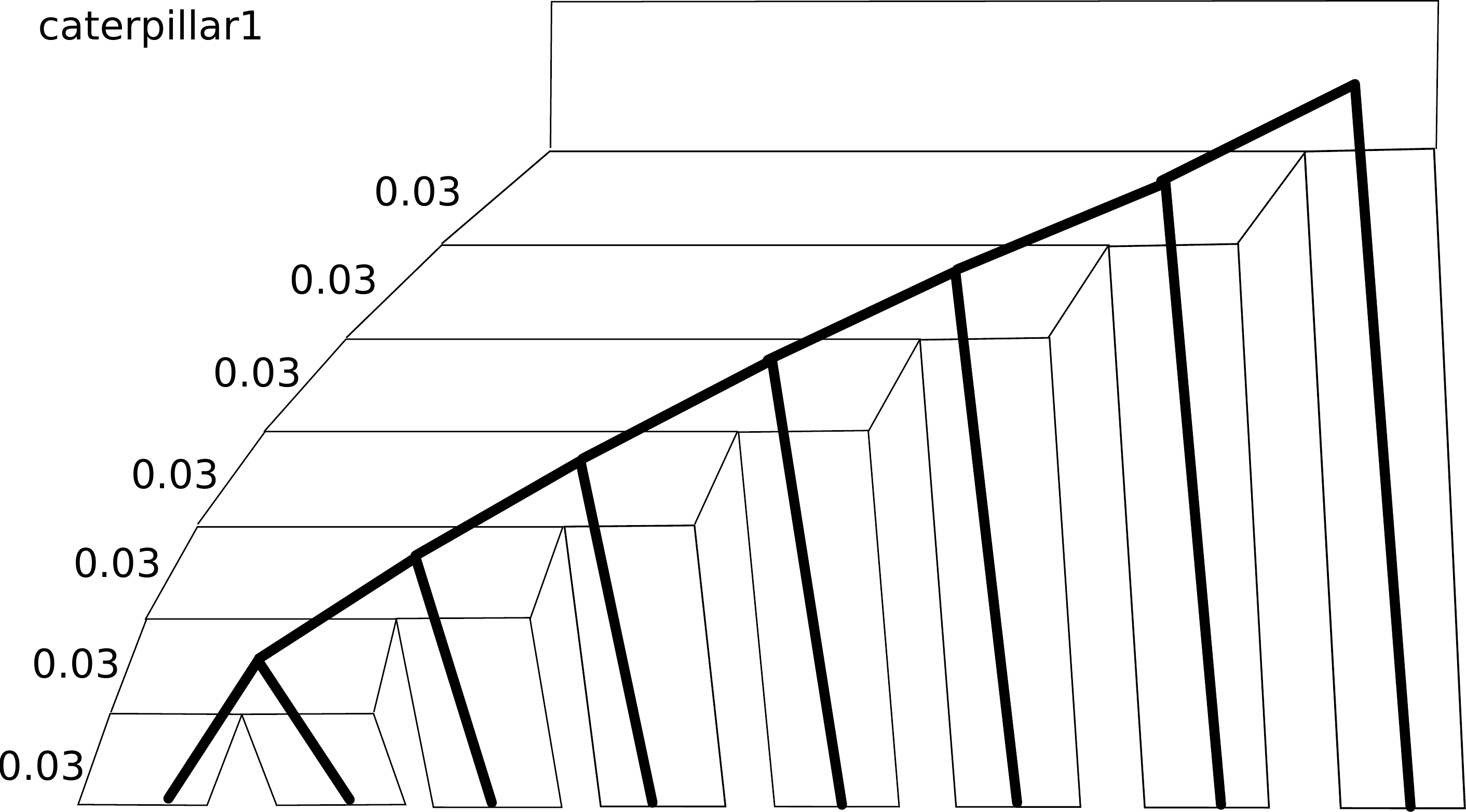}
\end{minipage}%
\hspace{0.5cm}
\begin{minipage}{.43\textwidth}
\centering
  \includegraphics[width=0.7\linewidth]{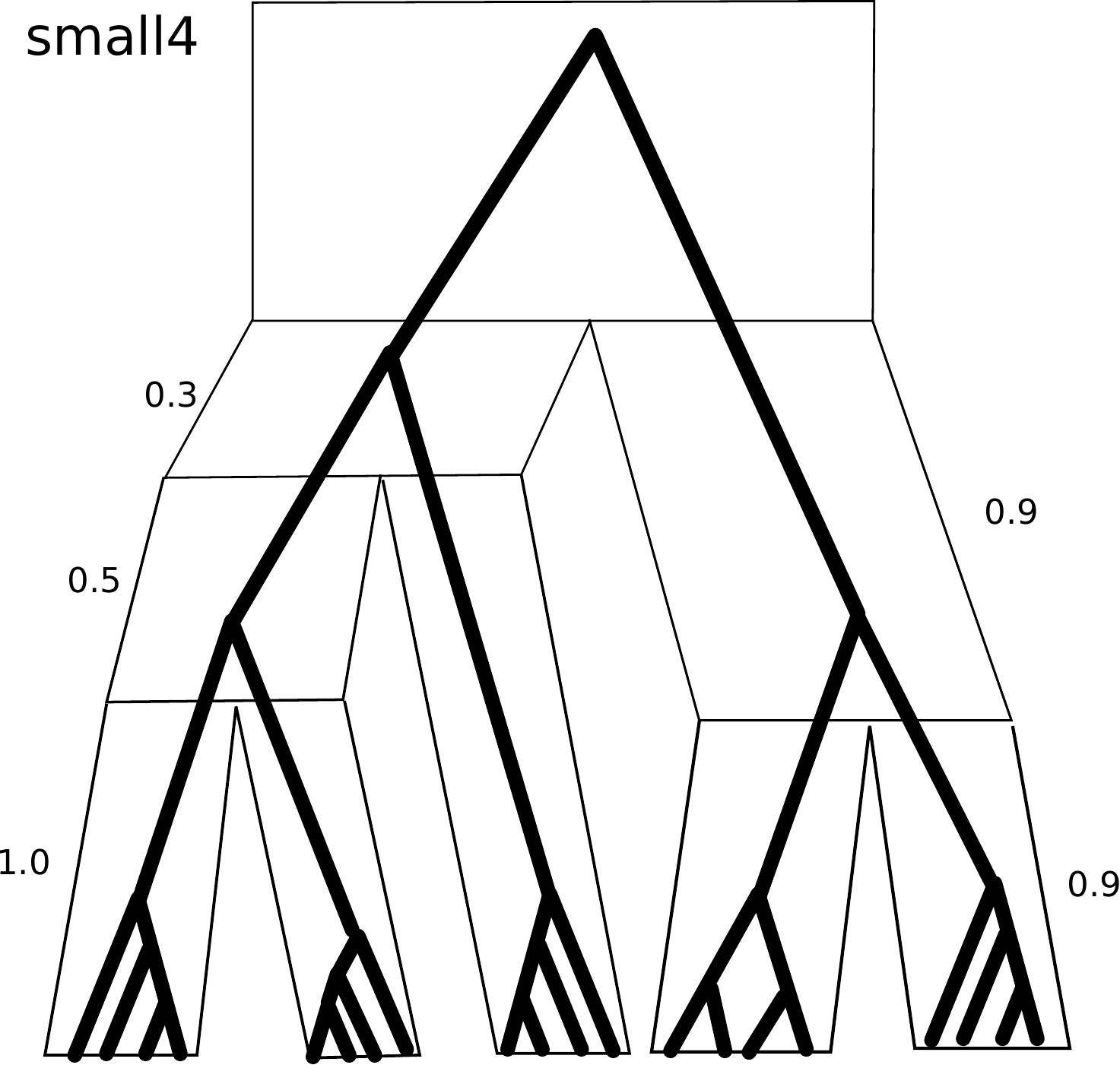}
  \\
  \vspace{0.5cm}
  \includegraphics[width=1.0\linewidth]{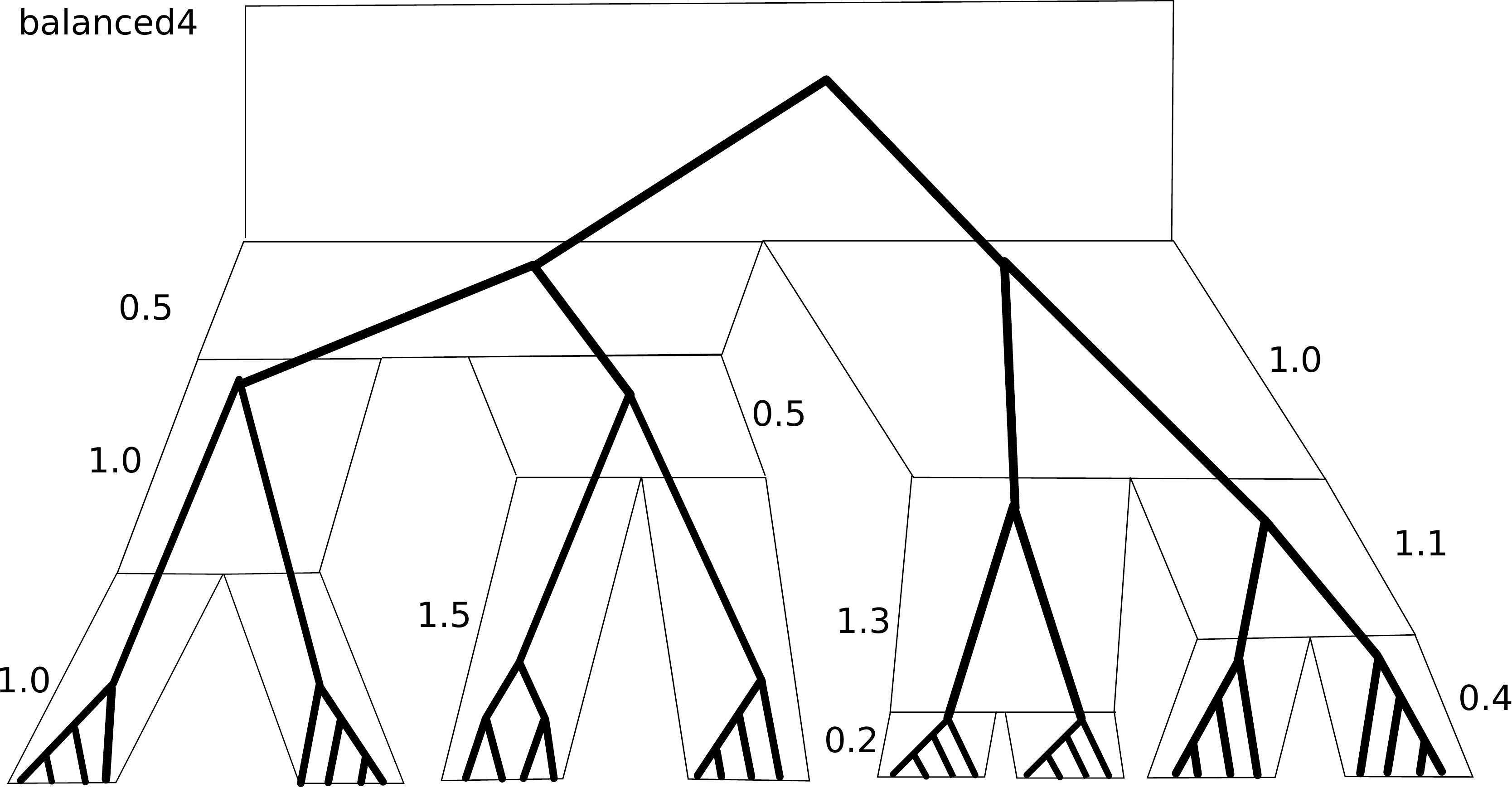}
  \\
  \vspace{0.5cm}
  \includegraphics[width=1.0\linewidth]{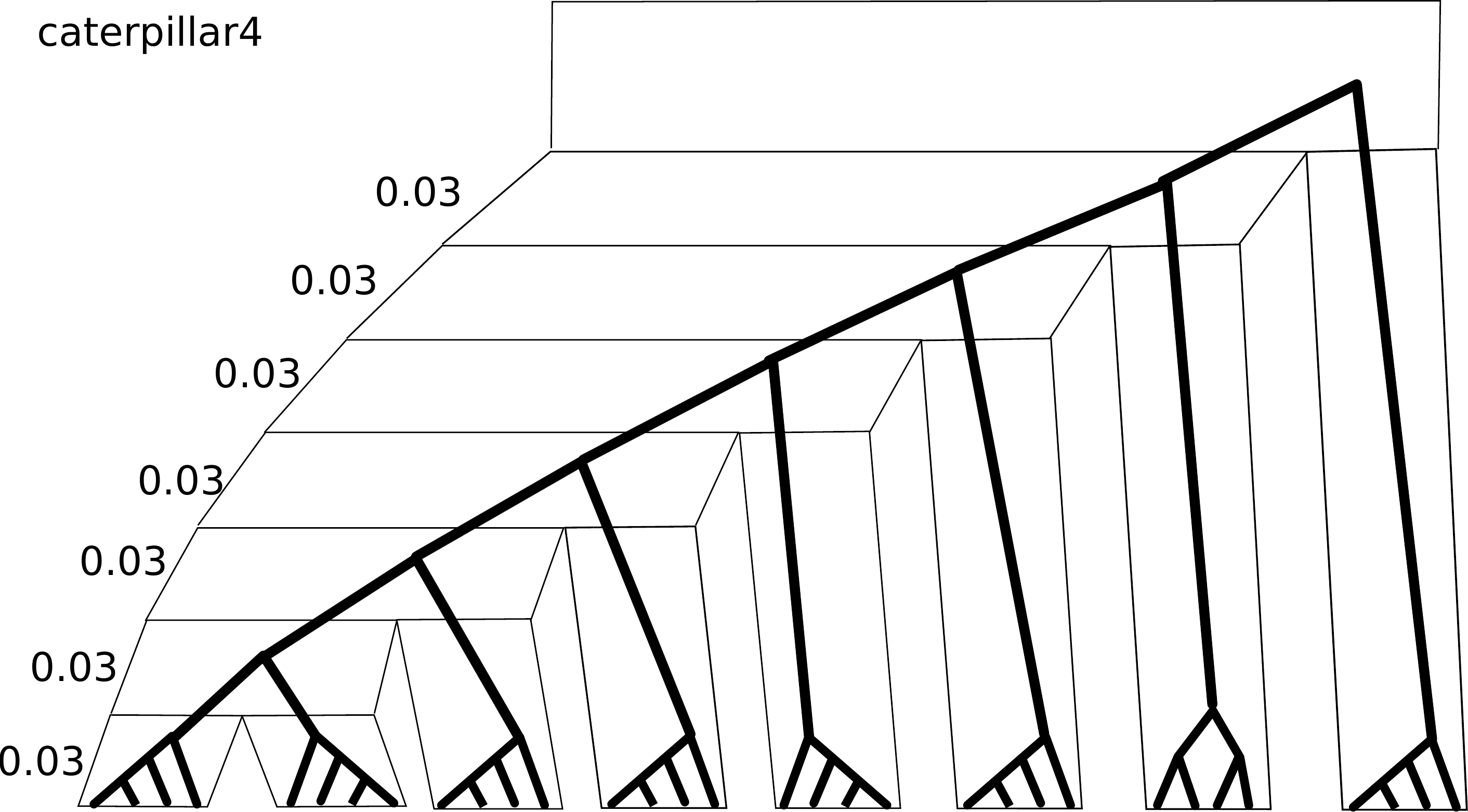}
\end{minipage}%
\caption{Data sets used in Table~\ref{tab-examples}. For each of the $3$ species trees, we generated two monophyletically concordant gene tree topologies with $1$ and $4$ genes in each species, respectively. Species tree branch lengths (in coalescent units) are indicated to the side of each species/parallelogram. All species trees are ultrametric, that is, all leaves are equidistant from the root.} 
\label{fig-examples}
\end{figure}

\clearpage

\section{Appendix: illustrative examples}\label{sec-examples}

To compare the behaviour of \textsc{gtprob} against that of STELLS and the gene-tree probability algorithm in the
\textsc{PhyloNet} package~\cite{yu2012probability}, we ran these programs on a few illustrative data sets, shown in Figure~\ref{fig-examples}.
The data sets {are obtained from} three species trees, where for each species tree we have generated gene trees with $1$ or $4$ genes sampled from each species. 
The results {for the resulting 6 data sets}, shown in Table~\ref{tab-examples}, demonstrate that the algorithms give virtually the same gene tree probabilities.
The higher discrepancy between STELLS and \textsc{gtprob} may be partially due to STELLS reporting log-probabilities 
with 10 significant digits, while the other two algorithms report about 16 significant digits.

\begin{table}[h]
    \centering
    \footnotesize
    \hspace*{-7mm}\begin{tabular}{|c|c|c||c|c|c||c|c|c|}
    Data set & sp. & g./sp. & \multicolumn{3}{c||}{Log-probability comparison} & \multicolumn{3}{c|}{Time (s)} \\
    \hline
    & & & \textsc{gtprob} & vs.~STELLS & vs.~\textsc{PhyloNet} & \textsc{gtprob} & STELLS & \textsc{PhyloNet} \\
    \hline
    small1 & 5 & 1 & $-1.644212722576583$ & $4 \cdot 10^{-9}$ & $2\cdot 10^{-14}$ & 0.34 & 0.00 & 1.08 \\
    small4 & 5 & 4 & $-17.814933301282963$ & $4\cdot 10^{-8}$& $4\cdot 10^{-14}$ & 0.50 & 0.07 & 2.03 \\
    balanced1 & 8 & 1 & $-7.518309941826544$ & $8 \cdot 10^{-9}$ & $3\cdot 10^{-14}$ & 0.49 & 0.00 & 1.17 \\
    balanced4 & 8 & 4 & $-61.25728682763628$ & $8\cdot 10^{-8}$& $5\cdot 10^{-14}$ & 0.56 & 6.01 & 43.61 \\
    caterpillar1 & 8 & 1 & $-11.822356365959706$ & $4 \cdot 10^{-8}$ & $6\cdot 10^{-14}$ & 0.43 & 0.00 & 1.19 \\
    caterpillar4 & 8 & 4 & $-67.26600152592931$ & $8\cdot 10^{-8}$& $1\cdot 10^{-11}$ & 0.55 & 1065.03 & 1164.64 \\
    \end{tabular}
    \caption{A comparison of different algorithms on the data sets in Figure \ref{fig-examples}.\label{tab-examples}
    Column ``sp.'' reports the number of species in the species tree, while column ``g./sp.'' reports the number of genes sampled per species.
    We report the log-probability returned by \textsc{gtprob} and the magnitude of the difference with log-probabilities returned by STELLS and \textsc{PhyloNet}, as well as running times. In data sets balanced1 and balanced4, the species tree is fully balanced. In data sets caterpillar1 and caterpillar4, the species tree is a caterpillar tree.}
\end{table}

Our algorithm is substantially faster than the other two approaches for data sets with
larger gene trees (balanced4 and caterpillar4).
The strong impact of the size of the gene tree on the running times of STELLS and \textsc{PhyloNet} is not surprising,
as each coalescence between lineages sampled from the same species $s$ can occur in any of the ancestors of $s$, leading to
a rapid increase in the number of evolutionary scenarios that these algorithms must enumerate.
This effect is particularly pronounced for the final data set, where the species tree is a caterpillar
(and where a given coalescence can occur in up to 8 different possible populations).
Comparing data sets balanced4 and caterpillar4 also shows that approaches based on enumeration can have very 
variable running times even for equally-sized trees. 
The running time of \textsc{gtprob} does not appear to be strongly impacted by the shape of the species tree.

In Section~\ref{sec:runtime}, we have remarked that our algorithm is likely to be faster for balanced species trees 
than for unbalanced trees in the limit of large data sets. For small data sets such as the ones in Table~\ref{tab-examples}, 
this effect may not discernible, but it is important to note that for larger data sets, this effect may at most increase by 1 the degree of the polynomial
bound on the running time of \textsc{gtprob}. Approaches based on enumeration, on the other hand, can go from being polynomial-time
for some tree shapes, to being exponential-time (or worse \cite{disanto2015coalescent}) for other tree shapes.
These observations are consistent with the differences in the variance of running times mentioned in Section~\ref{sec:runtime_moderate}.


\end{document}